\documentclass[onecolumn, conference]{IEEEtran}

\usepackage{amssymb}
\usepackage{amsmath}
\usepackage{epsfig}
\usepackage{cite}
\usepackage{enumerate}
\usepackage{multicol}
\usepackage{picinpar}

\newcommand{\define}{\stackrel{\triangle}{=}}

\newtheorem{theorem}{\bf Theorem}
\newtheorem{corollary}{\bf Corollary}
\newtheorem{lemma}{\bf Lemma}

\newcommand{\Kappa}{\mathcal{K}}

\begin{document}
\title{Interference Alignment and a Noisy Interference Regime for Many-to-One Interference Channels}

\author{\authorblockN{Viveck R. Cadambe, Syed A. Jafar}\\
\authorblockA{Center for Pervasive Communications and Computing\\Electrical Engineering and Computer Science\\
University of California Irvine, \\
Irvine, California, 92697, USA\\
Email: {vcadambe@uci.edu, syed@uci.edu}\\ \vspace{-1cm}}}

\maketitle
\vspace{10pt}
\begin{abstract} 
We study the capacity of discrete memoryless many-to-one interference channels, i.e., $K$ user interference channels where only one receiver faces interference. For a class of many-to-one interference channels, we identify a noisy interference regime, i.e., a regime where random coding and treating interference as noise achieves sum-capacity. Specializing our results to the Gaussian MIMO many-to-one interference channel, which is a special case of the class of channels considered, we obtain new capacity results. Firstly, while previous results characterized noisy interference regimes for many-to-one interference channels with inputs having average power constraints, we show that this remains valid for a more general class of inputs. This more general class of inputs includes the practical scenario of the inputs being restricted to fixed finite-size constellations such as PSK or QAM. Secondly, we extend noisy interference results previously studied in interference channels with single antenna nodes at all transmitters, to MIMO and parallel many-to-one interference channels. Finally, while previous results considered the Gaussian interference channel with full channel state information (CSI) at all nodes, we provide a noisy interference regime for fading Gaussian many-to-one interference channels without CSI at the transmitters.

While the many-to-one interference channel requires interference alignment, which in turn requires structured codes in general, we argue that in the noisy interference regime, interference is implicitly aligned by random coding irrespective of the input distribution. As a byproduct of our study, we identify a second class of many-to-one interference channels (albeit deterministic) where random coding is optimal (though interference is not treated as noise). We attribute the optimality of random coding in this second class of channels to the resolvability of the multiple interferers at the receiver which precludes the possibility of interference alignment and hence obviates the need of structured codes.
\end{abstract}

\section{Introduction}
The idea of \emph{interference alignment} has been recently discovered to play a significant role in the characterization of capacity of wireless interference networks \cite{MMK, Jafar_Shamai, Cadambe_Jafar_int,Bresler_Parekh_Tse}. Interference alignment is the idea that signals are designed so that they overlap at receivers where they cause interference while remaining distinguishable at receivers where they are desired.  In network communication scenarios where receivers face interference from multiple sources, alignment compacts the space occupied by the multiple interfering signals and results in increased rates for the messages desired at the receiver. Therefore, optimal code design for interference channels typically involves a conflict between the need for interference management via alignment at undesired receivers, and the need to maximize rates at the desired receiver. This conflict is clearly reflected in the contrast between interference channels, and channels that do not require alignment, viz. point-to-point, multiple access (MAC) and broadcast (BC) channels.  For instance, it is well known that the identically and independently distributed (i.i.d.) circularly symmetric Gaussian distributions on the inputs, which maximizes the differential entropy of the received signals, achieve capacity in the point to point, MAC and BC channels. However, in contrast, in interference channels, asymmetric complex signaling \cite{Cadambe_Jafar_Wang}, and structured (lattice) codes \cite{Cadambe_Jafar_Shamai, Bresler_Parekh_Tse, Sridharan_Jafarian_Vishwanath_Jafar, Etkin_Ordentlich, Motahari_DOFint} have been shown to be useful, especially because they align interference. In fact, the lack of a complete understanding of the limits of interference alignment is among the primary hurdles in capacity characterizations of wireless networks. In this paper, we will provide a finer understanding of interference alignment, and characterize the sum-capacity of a class of discrete memoryless (many-to-one) interference channels.

The results demonstrating the need for explicit interference alignment via lattice-coding/asymmetric complex signaling contrast with the ``noisy interference'' results for the Gaussian $K$ user interference channel found recently, presented in references \cite{Sreekanth_Veeravalli,Motahari_Khandani, Shang_Kramer_Chen}. These references showed that for the $K$ user interference channel, if the channel gains satisfy certain conditions, then, using \emph{random} codebooks with circularly symmetric Gaussian distributions for all messages and treating interference as noise at all receivers is sum-capacity optimal. Put differently, these results indicate that in certain scenarios, explicit alignment in the form of structured (lattice) coding or asymmetric signaling is not necessary, and random coding is optimal, even though there is potential for alignment with receivers facing multiple interferers. One of the main goals of this work is a better understanding of why random (Gaussian) coding is optimal in the noisy interference regime, in spite of the opportunity for alignment. We study this question in the setting of the discrete memoryless many-to-one interference channel - the interference channel where only one receiver faces interference - which is the simplest setting where a receiver faces multiple interferers. The main result of this work is the characterization of a noisy interference regime for a class of discrete memoryless many-to-one interference channels. The noisy interference condition obtained here can be loosely described as follows: \\\emph{ In the many-to-one interference channel, if the effective interference (with noise) seen by the only receiver facing interference is a stochastically degraded version of the set of received signals at all other receivers, then, random coding at all the transmitters and interference being treated as noise at the receiver facing interference achieves sum-capacity.}\\ \noindent The above result, which will be expressed rigorously later (Section \ref{sec:result}), holds for a broad class of discrete memoryless many-to-one interference channels including the \emph{Gaussian} many-to-one interference channel. From the perspective of Gaussian interference channels, we make two observations. Firstly, our main result captures the noisy interference regime for the single-antenna Gaussian many-to-one interference channel found previously in references \cite{Sreekanth_Veeravalli, Motahari_Khandani, Shang_Kramer_Chen}. Secondly, while results of references \cite{Sreekanth_Veeravalli, Motahari_Khandani, Shang_Kramer_Chen} are mainly restricted to the single-antenna Gaussian interference channels with \emph{classical} assumptions on the model, such as full channel state information (CSI) at all nodes and average power constraints on the inputs, our result described above holds for a broader class of discrete memoryless many-to-one interference channels, and is therefore more robust to the system model. Specializing our main result to the Gaussian setting enables us to extend the noisy interference results to scenarios of practical importance not captured by such classical assumptions and thus not previously considered. Before summarizing such extensions, we first describe how our main result summarized (in italics) above captures the noisy interference regime for many-to-one interference channels previously discovered in \cite{Sreekanth_Veeravalli,Motahari_Khandani,Shang_Kramer_Chen}. 

Consider a $K$ user Gaussian many-to-one interference channel (Fig. \ref{fig:gaussian}), whose inputs and outputs can be expressed as
\begin{figure*}
\begin{center}
\epsfig{file=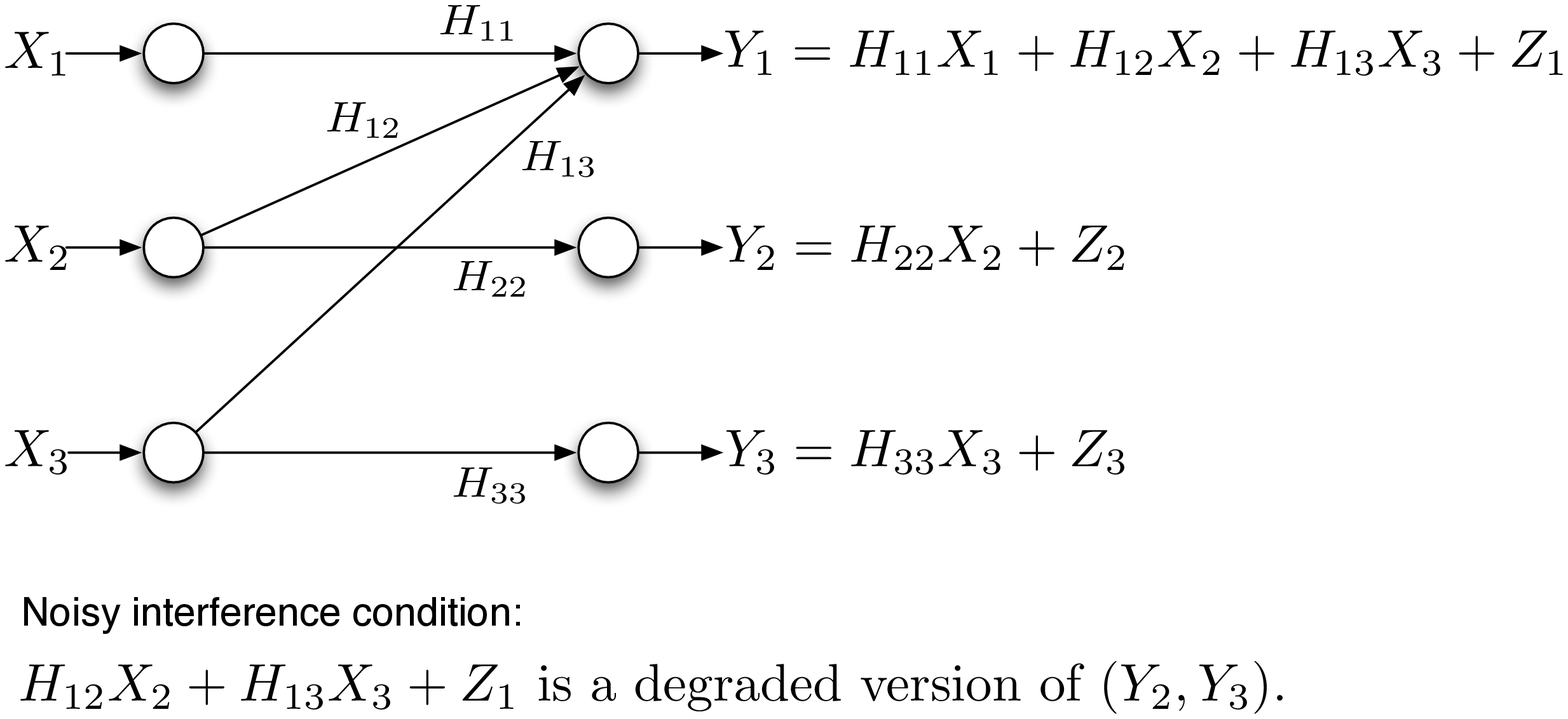, height=2.3 in, width=5 in}
\caption{The $3$-user Gaussian Many-to-one Interference Channel and the Noisy Interference Condition.}
\label{fig:gaussian}
\end{center}
\end{figure*}
\begin{eqnarray*}
{Y}_i(\tau)&=&{H}_{ii} {X}_i(\tau) + {Z}_i(\tau), i =2,3,\ldots,K \\ 
{Y}_1(\tau)&=& \sum_{i=1}^{K} {H}_{1i} {X}_i(\tau) + {Z}_1(\tau), i \in \Kappa_1 
\end{eqnarray*}
where, corresponding to the $\tau$th symbol, $X_i(\tau)$ is the complex scalar input at Transmitter $i$, $Y_j(\tau)$ is the complex scalar output at Receiver $j$ and $Z_{j}(\tau)$ represents the zero-mean unit-variance circularly symmetric additive white Gaussian noise (AWGN) variable at Receiver $j$. $H_{ji}$ is a complex scalar representing the channel gain between Transmitter $i$ and Receiver $j$. As is standard in the interference channel, Transmitter $i$ has a message to Receiver $i$, which is independent of the message at, and unknown to other transmitters (and unknown to all receivers prior to communication). Then, in this channel, with an average power constraint on the inputs, it is shown in \cite{Sreekanth_Veeravalli, Motahari_Khandani, Shang_Kramer_Chen}, that if 
\begin{equation}\sum_{j=2}^{K} \frac{|H_{1j}|^2}{|H_{jj}|^2} \leq 1,\label{eq:classicalnoisy} \end{equation}
then, circularly symmetric Gaussian inputs and treating interference as noise at all receivers is sum-capacity optimal. We note that the condition in (\ref{eq:classicalnoisy}) is a special case of the conditions stated in our main result above (in italics), i.e., when (\ref{eq:classicalnoisy}) holds, the effective interference at Receiver $1$, $V=\sum_{j=2}^{K} H_{1j}X_j+Z_1$, is a stochastically degraded version of the set of signals received at all other receivers, i.e., $(Y_2,Y_3,\ldots,Y_K)$. This is because $V$ is a degraded version of $\sum_{i=2}^{K}\frac{H_{1j}}{H_{jj}} Y_j$, which is, obviously, a degraded version of $(Y_2,Y_3,\ldots,Y_K)$. Thus, we have shown that, the noisy interference regime of \cite{Sreekanth_Veeravalli, Motahari_Khandani, Shang_Kramer_Chen} is included in the noisy interference regime found in our main result. In Appendix \ref{app:equivalence}, we show that, the two regimes - the regime described in (\ref{eq:classicalnoisy}) and the regime described by our main result - are in fact equivalent. While our results are equivalent\footnote{We have so far only discussed the optimality of random coding and treating interference as noise in our main result. References \cite{Sreekanth_Veeravalli,Motahari_Khandani, Shang_Kramer_Chen} also show the optimality using the circularly symmetric Gaussian distribution in the noisy interference regime; this optimality will be shown for Gaussian channels in our characterization as well in a formal description of our result in Section \ref{sec:result}.} to the results of previous works in the context of the classical single-antenna Gaussian many-to-one interference channels, as mentioned earlier, our result holds for a more general class of Gaussian many-to-one interference channels. Specifically, our results extend the noisy interference regime to many-to-one interference channels beyond the classical assumptions. We summarize such extensions below.

\begin{itemize}
\item Previous works \cite{Sreekanth_Veeravalli, Motahari_Khandani, Shang_Kramer_Chen} consider Gaussian interference channels where the input alphabet is continuous and there is an average power constraint on the input codewords. For these channels, the references show that for certain values of channel gains, random Gaussian codebooks and treating interference as noise is optimal when there is an average power constraint on the input codewords. In practice, however, input signals are typically restricted to fixed finite-size constellations such as PSK, QAM etc. It is not clear whether the noisy interference results results carry forward to the more practical setting of the inputs being constrained to fixed constellations. In fact, there remained open the question of whether there even exists a non-trivial set of channel gains where random coding and treating interference as noise achieves sum-capacity in this setting. In this work, we settle this open question by showing that the noisy interference regime remains valid in the Gaussian many-to-one interference channel even if the inputs are restricted to fixed constellations. In other words, if the channel gains satisfy the conditions of (\ref{eq:classicalnoisy}), then, random codebooks generated i.i.d with the appropriate distribution at the inputs, and treating interference as noise at the receiver facing interference achieves sum-capacity - even if the inputs are restricted to fixed constellations. Therefore, in the noisy interference regime, the capacity characterization problem is essentially reduced to the problem of determining the optimal \emph{single-letter} distribution on the inputs. 
\item The results of \cite{Sreekanth_Veeravalli, Motahari_Khandani, Shang_Kramer_Chen} are for interference channels with a single antenna at each node - the question of the existence and characterization of noisy interference regimes for MIMO interference channels remains open. In this work, we (partially) address this open question by characterizing a noisy interference regime for the MIMO Gaussian many-to-one interference channel. Note that extensions of \cite{Sreekanth_Veeravalli,Motahari_Khandani,Shang_Kramer_Chen} have been proposed to \emph{two} user MIMO interference channels \cite{Annapureddy_Veeravalli_MIMO, Shang_Kramer_Chen_Poor_MIMO}. Our result differs from the result of \cite{Annapureddy_Veeravalli_MIMO, Shang_Kramer_Chen_Poor_MIMO} in that, we present a noisy interference regime for the $K$-user interference channel, albeit not fully connected (since we only consider the many-to-one interference channel). It must be noted that the noisy interference regime for the MIMO setting also remains valid for input signals being restricted to finite constellations.
\item Previous noisy interference results are presented for the case where the channel is constant (i.e., not fading), and when transmitters and receivers have channel state information (CSI). In this paper, we obtain a noisy interference regime for the fading Gaussian many-to-one interference channel where transmitters do not have CSI, and only the receivers have CSI. 
\item Previous results \cite{Sankar_inseparable, Cadambe_Jafar_inseparable} have shown that \emph{parallel} (i.e. multi-carrier) Gaussian interference channels (including many-to-one interference channels), unlike point-to-point, multiple access and broadcast channels, are in general \emph{inseparable}, i.e., joint coding over the multiple carriers is required to achieve sum-capacity in parallel interference channels. While parallel interference channels are in general inseparable, \emph{under certain special conditions}, they are separable, i.e., separate (independent) coding over the various carriers (and in fact, treating interference as noise) achieves sum-capacity. Such conditions have been identified for parallel single-antenna $Z$ interference channels in \cite{Sankar_etal_separability} and for the (fully-connected) $2$-user Gaussian interference channels in \cite{Shang_Kramer_Chen_Poor_parallel}. In this paper, we extend the results of \cite{Sankar_etal_separability} to Gaussian MIMO many-to-one interference channels. In particular, we show that, under the special case that the many-to-one interference channels formed over \emph{each} of the carriers forming the parallel channel satisfies our noisy interference conditions, the channel is separable from a sum-capacity perspective. For example, in the single-antenna Gaussian many-to-one interference channel, if the channel gains on each of the carriers satisfy the condition of (\ref{eq:classicalnoisy}), then separate random coding and treating interference as noise achieves sum-capacity. Therefore, in this case, with an average power constraint on the input, the sum-capacity of the parallel many-to-one interference channel is the sum of the capacities of the various individual carriers under an optimal power allocation - much like the point-to-point, MAC and BC channels. Further, this separability result is not limited to the average power constraint on the inputs, and holds even for inputs of fixed finite constellations. It must be noted that our main result automatically implies that random coding and treating interference as noise over such a channel (where each carrier satisfies our noisy interference criterion) achieves sum-capacity, because the required degradedness condition holds for the parallel channel. But our main result described above does not, however, imply their separability - the property that the optimal distribution used in random coding has the input over each carrier independent of the input of the other carriers. The separability is an additional result shown in Section \ref{sec:separable}.
\end{itemize}

\subsection*{Why is explicit interference alignment not required in noisy interference regimes?}
An important insight to emerge from this work is that in the noisy interference regime, interference is aligned implicitly via random codes. The idea of interference alignment with random codes can be understood in the following setting. Consider a receiver receiving multiple signals coded from a codebook generated in the classical random coding fashion. If the cardinalities of the codebooks corresponding to these signals lie in the achievable random coding rate region (with the corresponding input distributions) of the multiple access channel formed at the receiver, then the receiver can resolve these multiple signals with high probability. In other words, the signals are not aligned. On the other hand, if the cardinalities of the codebooks lie outside this achieved rate region of the multiple access channel formed at the receiver, then the signals align. In fact, in this scenario, the signals cannot be resolved uniquely at the receiver, \emph{because} the signals align. While alignment is not a desirable phenomenon if the receiver intends to resolve the signals as is the case in the multiple access channel, it is beneficial if the signals are interfering at the receiver as is the case in interference channels. In the noisy interference regime for the many-to-one interference channel, we show that because of a degraded nature of the channel, interference can be aligned with random codes for any distribution on the inputs. In particular, in the Gaussian channel, interference is aligned, even with random coding and with the circularly symmetric Gaussian distribution; thus, alignment is implicit in this case.

It must be noted that the optimality of random coding in the noisy interference regime is desirable from two perspectives. First, the generality of random coding argument enables us to present results for a fairly broad class of channels which may or may not be linear (though we later {specialize} our results to the linear Gaussian setting). Secondly, the optimality of random codes enables a single-letter characterization for the capacity, unlike in channels which need structured codes, where single-letter characterizations may not even exist \cite{Philosof_Zamir}. 

The idea of implicit interference alignment in the noisy interference regime is examined more closely in Section \ref{sec:deterministic} by specializing the noisy interference regime to a \emph{deterministic} many-to-one interference channel. In this setting, we observe that if a random code transmits at sufficiently high rates, then the interference becomes noisy and the extent of alignment via random codes is optimal. The simpler setting of the deterministic channel, apart from enabling a better understanding of the idea of alignment via random coding, allows two other interesting insights into interference alignment. Firstly, we find that random coding achieves capacity in a scenario where the multiple interferers are resolvable at the receiver facing interference (The idea of resolvable interference has been earlier used to determine the capacity of a class of symmetric deterministic interference channels in \cite{Gou_Jafar}). The resolvability of interference precludes the possibility of interference alignment which enables a characterization of its capacity region. Secondly, a combination of insights from the noisy interference regime and the resolvable interference regime enables us to provide, in Section \ref{sec:discussion1}, a (partial) answer to the question : How many bits of additional rate can interference alignment provide on the many-to-one interference channel? 

We now proceed to the next section where we formally define the discrete memoryless many-to-one interference channel - the basic setting of all the results of this paper.

\section{System Model : A Class of Discrete Memoryless Many-to-One Interference Channels}
\label{sec:sysmod}

\begin{figure*}
\begin{center}
\epsfig{file=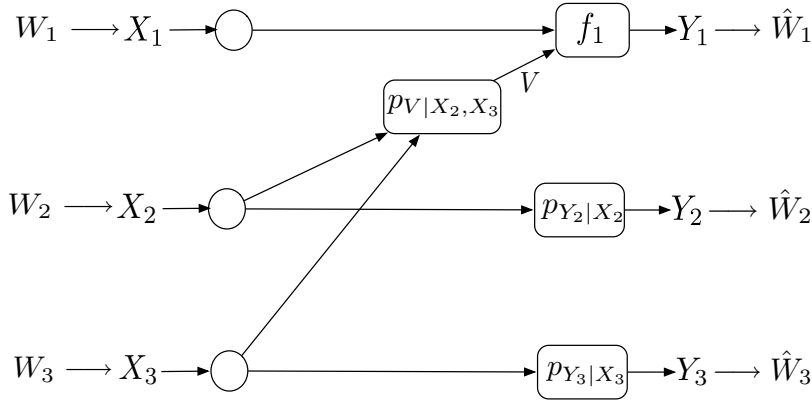, height=2.1 in, width=4.2in}
\caption{The $3$-user Many-to-one Interference Channel}
\label{fig:many-one}
\end{center}
\end{figure*}

The $K$ user discrete memoryless many-to-one interference channel (Figure \ref{fig:many-one}) is defined by a set of $K$ inputs $X_i \in \mathcal {X}_i$, and a set of $K$ outputs $Y_i \in \mathcal{Y}_i$ for $i=1,2,\ldots,K$. In the class of many-to-one interference channels considered, the outputs $Y_{i}, i=2,3,\ldots, K$ are generated using the distributions $p_{Y_{i}|X_i}$. The output $Y_1$ is generated as 
$$ Y_1 = f_1(X_1,V),$$
where $V \in \mathcal{V}$ is generated using $p_{V|X_2,X_3,\ldots,X_K}$. We assume that $V$ is invertible from $(Y_1,X_1),$ i.e., there exists a function $f_1^{-1}$ such that 
\begin{equation}V=f_1^{-1}(Y_1,X_1). \label{eqn:invertible}\end{equation}
There are $K$ independent messages, with message $W_{i} \in \mathcal{W}_i$ generated at source $i\in \{1,2,\ldots, K\}$, with each message being uniformly distributed over the corresponding message set.
 A \emph{code} of length $T$ symbols consists of encoding functions (or equivalently, codebooks) $\phi_i:\mathcal{W}_i \rightarrow \mathcal{X}_i^T$ and decoding functions $\psi_i:\mathcal{Y}_i^T \rightarrow \mathcal{W}_i$ for all $i \in \{1,2,\ldots K\}$. It is assumed that all the codebooks, i.e., all the mappings $\phi_i,i=1,2,\ldots,K,$ are known to all the decoders. We restrict our study to channels and constraints on codewords which ensure that one of the following two sets of quantities exist
\begin{itemize} 
\item $H(Y_i^T)$ and $H(Y_i^T|X_i^T)$ exist for $i=1,2,\ldots,K$. Note that using $i=1$, and (\ref{eqn:invertible}), this automatically implies that $H(V^T)$ exists. Also note that the this condition captures all channels where the alphabets $\mathcal{X}_i,\mathcal{Y}_i$ are finite.
\item $h(Y_i^T)$ and $h(Y_i^T|X_i^T)$ exist for $i=1,2,\ldots,K$. Note that using $i=1$, and (\ref{eqn:invertible}), this automatically implies that $h(V^T)$ exists.
\end{itemize} 

The average probability of error of the code $P_{e}^{(T)}$ is defined to be the probability that the set of decoded messages is not identical to the set of encoded messages, i.e., $$ P_{e}^{(T)}\define\Pr\left(\left\{ \exists i \in \{1,2,\ldots,K\}| \psi_i(Y_i^T) \neq W_i\right\}\right).$$ The rate of the code is the tuple $\vec{R}=(R_1,R_2,\ldots,R_K)$, where $R_{i}=\frac{\log|\mathcal{W}_i|}{T}$, with $|\mathcal{W}_i|$ denoting the cardinality of the message set $\mathcal{W}_i$.  A rate-tuple $\vec{R}$ is said to be \emph{achievable} if there exists a sequence of codes, all of rate $\vec{R}$, such that average probability of error vanishes asymptotically, as the sequence index increases. Let $\mathcal{C}$ be the closure of the set of all achievable rate tuples. The sum-capacity $C_{\Sigma}$ of the interference channel is defined as 
$$ C_{\Sigma} \define \max_{(R_1,R_2,\ldots,R_K) \in \mathcal{C}} \sum_{i=1}^{K} R_i.$$

\subsubsection*{Notation}
We use the notation $A^T$ to denote $(A(1),A(2),\ldots,A(T))\in\mathcal{A}^T$ for any random variable $A$. The calligraphic notation is used to indicate sets. The notation $\mathcal{N}(\mu, \Lambda)$ is used to indicate a circularly symmetric complex Gaussian random vector with mean $\mu$ and covariance matrix $\Lambda$. $I_{N}$ is used to denote the $N \times N$ identity matrix. The following quantities are also used in the paper. 
\begin{eqnarray*}
\Kappa_1&=&\{2,3,\ldots,K\}\\
\Kappa&=& \Kappa_1 \cup \{1\}\\
Y_{\mathcal{A}} &=& \{Y_i, i \in \mathcal{A}\}, \mathcal{A} \subseteq \Kappa.
\end{eqnarray*}

Before we proceed, the following points must be noted.
\begin{itemize}
\item The channel studied here is a natural adaptation, to the many-to-one interference channel setting, of the $2$ user interference channel studied in \cite{Telatar_Tse}.
\item In the special case where all the alphabets $\mathcal{X}_i,\mathcal{Y}_i, \mathcal{V}$ are finite for all $i$, and all the distribution functions are deterministic (i.e. $H(Y_i|X_i)=H(V|X_2,X_3,\ldots,X_K)=0$ for all $i=2,\ldots,K$ irrespective of the input distribution), the channel is an adaptation, to the many-to-one interference channel setting, of the class $2$ user deterministic interference channels studied by El Gamal and Costa in \cite{Gamal_Costa}. 
\item As we describe next, the MIMO Gaussian many-to-one interference channel is a special case of the class of channels described above.
\end{itemize}
\subsection*{The MIMO Gaussian Many-to-One Interference Channel}
Consider the MIMO Gaussian interference channel with $M_i$ antennas at Transmitter $i$ and $N_i$ antennas at Receiver $i$ so that $\mathcal{X}_i \subseteq \mathbb{C}^{M_i}$, $\mathcal{Y}_i = \mathbb{C}^{N_i}$ 
\begin{eqnarray}
{Y}_i(\tau)&=&{H}_{ii} {X}_i(\tau) + {Z}_i(\tau), i \in \Kappa_1 \label{eq:Gaussian1}\\ 
{Y}_1(\tau)&=& \sum_{i \in \Kappa} {H}_{1i} {X}_i(\tau) + {Z}_1(\tau), i \in \Kappa_1 \label{eq:Gaussian2}
\end{eqnarray}
where, corresponding to the $\tau$th symbol, ${X}_i$ is the $M_i \times 1$ vector representing the input at Transmitter $i$, ${Y}_i, {Z}_i$ are the $N_i \times 1$ vectors representing the output and the additive white Gaussian noise vectors at Receiver $i$. We assume that all the noise vectors are circularly symmetric with zero mean and a covariance matrix of identity. This channel can be reduced to the channel defined previously (Figure \ref{fig:many-one}) if we set 
\begin{eqnarray*} 
V&=&\sum_{i=2}^{K}H_{1i} X_i+Z_1,\\
f(X_1,V)&=&H_{11} X_1+V,
\end{eqnarray*}
where $V$ is a $N_1 \times 1$ vector. Note that since the constraints on the inputs $X_i$ are fairly general, we capture most scenarios of interest such as inputs from a finite fixed-size constellation (Example : PSK, QAM) and inputs with an average and power constraints on the codeword. Further, for the special case where $\mathcal{X}_i=\mathbb{C}^{M_i}$ and an average power constraint is imposed on the input codewords, we will give an explicit expression of the capacity of the channel in the noisy interference regime (to be defined later) in terms of the powers $P_i$, where 
$$ E\left[\frac{1}{T}\sum_{\tau=1}^{T} || X_i(\tau)||^2\right] \leq P_i,$$
and $T$ denotes the length of the codeword.

\section{A Noisy Interference Regime for Many-to-One Interference channels}
\label{sec:result}

Before we proceed to the main result, we introduce a lemma which is useful in the proofs.
\begin{lemma}
\label{lemma:extremal}
{\it
Consider random sequences $A^{T},B^{T},C^{T}$ such that $A \in \mathcal{A}, B \in \mathcal{B}, C \in \mathcal{C}$, and $B^T,C^T$ are generated as $p(B^T,C^T,A^T)=p(A^T)\displaystyle\prod_{\tau=1}^{T}p(B(\tau)|A(\tau))p(C(\tau)|B(\tau))$, where $p(B(\tau)|A(\tau))=p_{B|A}, p(C(\tau)|B(\tau))= p_{C|B}$, for all $\tau=1,2,\ldots,T$. Note that the sequence $A^{T}$ does \emph{not} have to be generated in an i.i.d fashion. The sequences $B^{T},C^{T}$ can be interpreted to be outputs of a physically degraded discrete memoryless broadcast channel whose input sequence is $A^T$. Also note that $A^{T} \rightarrow B^{T} \rightarrow C^{T}$. Now, if $H(B^{T}),H(C^{T})$ exist, then 
$$H(B^{T})-H(C^{T}) \leq \sum_{\tau=1}^{T} \left(H(B(\tau))-H(C(\tau))\right)$$
If $h(B^{T}),h(C^{T})$ exist, then 
$$h(B^{T})-h(C^{T}) \leq \sum_{\tau=1}^{T} \left(h(B(\tau))-h(C(\tau))\right)$$
Further, suppose $\mathcal{A} = \mathbb{C}^{M}$,  $p_{B|A} \thicksim \mathcal{N}(0,{\Lambda}_1)$ and $p_{C|B}\thicksim \mathcal{N}(0,{\Lambda}_2)$, or equivalently, let there be variables $Z_i, i=1,2$ generated i.i.d according to $Z_i \thicksim \mathcal{N}(0, \Lambda_i), i \in \{1,2\}$, and $\forall \tau \in \{1,2,\ldots,T\}$
\begin{eqnarray*}
B(\tau)&=&A(\tau)+Z_1(\tau)\\
C(\tau)&=&B(\tau)+Z_2(\tau)
\end{eqnarray*}
 Also, let consider a covariance matrix constraint on the sequence $A^T$, i.e.,
$$ E\left[\frac{1}{T}\sum_{\tau=1}^{T} A(\tau)A(\tau)^{\dagger}\right] = \Gamma$$ for some covariance matrix $\Gamma$, then i.i.d circularly symmetric Gaussian distribution on $A^{T}$ maximizes the quantity $h(B^{T})-h(C^{T})$.}
\end{lemma}
The proof of the lemma is placed in Appendix \ref{app:extremal}. The reader may note that a special case of the above lemma, where $\mathcal{A}=\mathcal{B}=\mathcal{C}=\mathbb{C}$ and $p_{B|A}$, $p_{C|B}$ are both Gaussian, is used in showing previous noisy interference results for SISO interference channels \cite{Sreekanth_Veeravalli, Motahari_Khandani, Shang_Kramer_Chen}.
We now present a noisy interference regime for the discrete memoryless many-to-one interference channel considered in this paper.
\begin{theorem}
\label{thm:noisy}
{\it 
In the many-to-one interference channel defined previously, if $p_{V|X_{\Kappa_1}}$ is a degraded form of $p_{Y_{\Kappa_1}|X_{\Kappa_1}}=\prod_{i\in\Kappa_1} p_{Y_i|X_i}$ then, its sum-capacity can be achieved with random coding and treating interference as noise, and it can be expressed as 
\begin{equation} C_{\Sigma} = \max\sum_{i\in \Kappa} I(X_i;Y_i),\label{eq:sumcapacity}\end{equation}
where the maximization is carried over all probability distributions on the input which factorize as $\prod_{i=1}^{K} p_{X_i}$. }
\end{theorem}
The theorem is proved in Appendix \ref{app:noisy}. We now apply the above result to the Gaussian setting below.
\begin{corollary}
\label{cor:Gaussian}
{\it
Consider the Gaussian many-to-one interference channel as defined by (\ref{eq:Gaussian1})-(\ref{eq:Gaussian2}).
If, for $i \in \Kappa_1$, there exist covariance matrices $\Lambda_i$ so that \begin{itemize}
\item $\sum_{i\in \Kappa_1} \Lambda_i \prec {I}_{N_1}$ and
\item  $U_{i} \define H_{1i}X_i + \tilde{Z}_i$ is a stochastically degraded version of $Y_{i}$, where $\tilde{Z}_{i}\thicksim \mathcal{N}(0, \Lambda_i)$, 
\end{itemize}
then, a single-letter input distribution on inputs with Receiver $1$ treating all interference as noise is sum-capacity optimal. Further, if $\mathcal{X}_i = \mathbb{C}^{M_i}$ with a power constraint $P_i$ on the input codeword $X_i^T$, then the optimal input distribution on the input is Gaussian, i.e., the capacity of the channel can be expressed as 
$$ C_{\Sigma}=\max_{\Gamma_i, \textrm{tr}(\Gamma_i) \leq P_i, i \in \Kappa} \sum_{i=1}^{K} \log\left(\frac{\det\left(I_{N_i}+\displaystyle\sum_{k\in \Kappa}H_{ik}\Gamma_k H_{ik}^{\dagger}\right)}{\det\left( I_{N_i}+\displaystyle\sum_{k \in \Kappa-\{i\}}H_{ik}\Gamma_k H_{ik}^{\dagger}\right)}\right)$$ with $H_{ij}=0$ if $i \notin \{1,j\}$. }
\end{corollary}
The proof is almost identical to Theorem \ref{thm:noisy}, on noting that $V$ is a degraded version of $U_{\Kappa_1}=(U_2,U_3,\ldots,U_K)$, which is in turn, a degraded version of $Y_{\Kappa_1}$, effectively implying that $V$ is a degraded version of $Y_{\Kappa_1}$. For completeness, we provide the proof in Appendix \ref{app:Gaussian}.
\subsubsection*{Remark}
In this case, if the inputs come from a finite constellation such as PSK/QAM, then the capacity characterization essentially involves determination of the optimal input distribution which maximizes (\ref{eq:sumcapacity}). Unlike the point-to-point channel, it is not clear even for symmetric constellations such as BPSK, whether the optimal distribution is uniform on the inputs.
\subsubsection*{Remark}
The noisy interference condition of Theorem \ref{thm:noisy} and the corresponding sum-capacity characterization remains unchanged, even if Transmitter $i \in \Kappa_1$ each had an independent message for Receiver $1$, along with the usual message for Receiver $i$, so that there are $2K-1$ messages in the system - i.e., even if each link of the many-to-one channel carried a message. In this channel, with the channel satisfying the conditions of the theorem, all the messages to Receiver $1$ from Transmitter $i \neq 1$ will be set to null so that the channel operates as an interference channel, for sum-capacity. The proof is almost identical as for the interference channel, with minor adaptations which are demonstrated in \cite{Huang_Cadambe_Jafar}; the reference showed that the noisy interference regime derived for the two-user interference channel in \cite{Sreekanth_Veeravalli,Motahari_Khandani,Shang_Kramer_Chen} remains unchanged even if each transmitter had a message to all receivers to form the two-user $X$ channel.

\subsection{Examples}

\subsubsection*{Example 1 - SIMO Gaussian Many-to-one Interference Channel}
The above theorem generalizes the noisy interference regime for many-to-one interference channels shown in \cite{Sreekanth_Veeravalli,Motahari_Khandani, Shang_Kramer_Chen} . To see this, consider the SIMO many-to-one interference channel where all the inputs are one-dimensional scalars, whereas all the outputs are vectors. In this case, note that the channel from Transmitter $i$ to Receiver $j$ can be represented by the vector $H_{ji}$. Without loss of generality, let us assume that $||H_{ii}||^2=1$. In this case, it can be verified that, if 
$$ \sum_{i\in \Kappa_1} ||H_{1i}||^2 \leq 1,$$
then treating interference as noise is optimal. This can be seen with the auxiliary variables $\tilde{Z}_i \thicksim \mathcal{N}(0,||H_{1i}||^2)$ and, as mentioned in the above corollary, $U_{i}=H_{1i} X_i+\tilde{Z}_i$. For the SISO case, the above condition boils down to the conditions specified in \cite{Sreekanth_Veeravalli,Motahari_Khandani, Shang_Kramer_Chen}.

\subsubsection*{Example 2 - MIMO Gaussian Many-to-one Interference Channel}
Consider a MIMO interference channel where $M_i=N_i=M, \forall i \in \Kappa_1$ and $N_1$=1. Since a MIMO $M \times M$ channel can be decomposed into $M$ parallel links using singular-value decomposition, we can assume that the channel matrices $H_{ii}, i \in \Kappa_1$ are diagonal without loss of generality. Let $H_{ii}^{(k)}$ denote the $k$th diagonal entry of $H_{ii}$. $H_{1i}, i \in \Kappa_1$ is a $M \times 1$ vector, whose $k$th entry, we denote by $H_{1i}^{(k)}$. Now, if 
$$ \sum_{i \in \Kappa_1} \sum_{k=1}^{M} \frac{|H_{1i}^{(k)}|^{2}}{{|H_{ii}^{(k)}|^{2}}} \leq 1, $$
then, treating interference as noise is optimal. This can be noted by setting $\tilde{Z}_{i}\thicksim \mathcal{N}\left(0,\sum_{k=1}^{M} \frac{|H_{1i}^{(k)}|^2}{|H_{ii}^{(k)}|^2}\right)$ and $U_{i}=H_{1i}X_i + \tilde{Z}_i$, for $i \in \Kappa_1$. 

\subsubsection*{Example 3 - Fading Gaussian Many-to-one Interference Channel without CSIT}
Consider a SISO Gaussian interference channel with Rayleigh fading. In this case, let the received signals maybe expressed similar to (\ref{eq:Gaussian1})-(\ref{eq:Gaussian2}), where all the quantities are scalars; the only difference being that, in this case, the channel fade $H_{ij}$ is time-varying. Specifically, the input-output relations can be expressed as 
\begin{eqnarray*}
{Y}_i(\tau)&=&{H}_{ii}(\tau) {X}_i(\tau) + {Z}_i(\tau), i \in \Kappa_1 \\ 
{Y}_1(\tau)&=& \sum_{i \in \Kappa} {H}_{1i}(\tau) {X}_i(\tau) + {Z}_1(\tau), i \in \Kappa_1,
\end{eqnarray*}
where, all the receivers have $1$ antenna so that the inputs $X_i$, outputs $Y_i$ and the channel gains $H_{ji}$ are scalars for all $i \in \Kappa_1$ and $j=1$ or $j=i$. We assume that $H_{ij}(\tau)$ is drawn i.i.d according to a circularly symmetric Gaussian distribution $\mathcal{N}(0,\sigma_{ij}^2)$. Note that this is the classical Rayleigh fading model, where the magnitude of the fade $|H_{ij}|$ is Rayleigh distributed with parameter $\sigma_{ij}$. We assume that $H_{ij}$ is independent of $H_{i^{'}j^{'}}$ if $i \neq i^{'}$ or $j \neq j^{'}$.  The receivers are aware of the channel state information, or equivalently, the effective output at Receiver $j$, corresponding to the $\tau$th symbol can be expressed as $(Y_{j}(\tau), \mathcal{H}(\tau))$, where $\mathcal{H}(\tau)=\{H_{ij}(\tau): i = 1, \textrm{ or }, i = j, i, j \in \Kappa, \}$. We assume that the transmitters do not have CSI, so that the input codewords are independent of the channel gains. Suppose $\sum_{i=1}^{K}\frac{\sigma_{1i}}{\sigma_{ii}} \leq 1$, then, the noisy interference condition is satisfied. This can be verified on noting that, if $\sigma_{ij}$ satisfy the specified condition, the effective interference at Receiver $1$, $(V, \mathcal{H}),$ is a degraded version of $(\sum_{i\in\Kappa_1} \frac{\sigma_{1i}}{\sigma_{ii}}Y_i, \mathcal{H}).$ Thus, this condition on the variances $\sigma_{ij}$ provides a condition for noisy interference in the fading Gaussian many-to-one interference channel without CSIT.

\subsubsection*{Example 4 - Collision-based Interference Channel Model}
Here, we construct a collision-based model for the many-to-one interference channel. Intuitively, the model can be explained as follows. Transmitter $i \in \Kappa_1$ has two choices - it can transmit a symbol from a finite set $\mathcal{X}_i^{'}$, or it can choose to remain silent. The signal transmitted by Transmitter $i \in \Kappa_1$ is received perfectly at Receiver $i$. At Receiver $1$, there are two possibilities: the signal transmitted by Transmitter $1$ can be received perfectly, or it can be erased (due to a collision). The probability of collision/erasure can be designed based on the set of transmitters which are silent. Formally, the model can be constructed as follows. Consider a $K$ user many-to-one interference channel, where $\mathcal{X}_i = \{\phi\} \cup \mathcal{X}_i^{'}$, where $\mathcal{X}_i^{'}$ is a finite set which does not contain the element $\phi$. The element $\phi$ is used to indicate the case where user $i$ remains silent (Note that this symbol can be used in the code at Transmitter $i$ to convey information to the Receiver $i$). The received signals are defined as $Y_i = X_i$ for $i \in \Kappa_1$, and $\mathcal{V}=\{0, \varepsilon\}$ with $$Y_1=f_1(X_1,V)=\left\{\begin{array}{cc}X_1,& \textrm{if }V=0\\ \varepsilon &\textrm{if }V=\varepsilon\end{array}\right\}.$$
Also, $V$ is drawn based on any probability distribution $p_{V|X_{\Kappa}}$. In this channel, $V=\varepsilon$ can be interpreted as an occurrence of a collision at Receiver $1$.  In this channel, clearly, $V$ is a degraded version of $Y_{\Kappa_1}$ and therefore, the search for the sum-capacity of this channel is reduced to the search of the optimal single-letter distribution on all the inputs. Note that this model captures the traditional collision based medium access models - this can be seen by setting $V=\varepsilon$ deterministically, when any of the users that 'collide' with user $1$ are transmitting any symbol other than $\phi$, and setting $V=0$ otherwise. With the optimal input distribution, since treating interference as noise is optimal at Receiver $1$, the receiver effectively observes a binary erasure channel, with the erasure probability calculated based on the input distribution and $p_{V|X_{\Kappa}}$.

\subsection{Separability of Parallel Noisy Discrete Memoryless Many-to-one Interference Channels}
\label{sec:separable}
In this section, we consider a parallel extension of the class of discrete memoryless many-to-one interference channels introduced in Section \ref{sec:sysmod}. Specifically, we show that if the many-to-one interference channel formed over \emph{each} carrier (parallel component) of the parallel channel satisfies the noisy interference condition of Theorem \ref{thm:noisy}, then \emph{separate} random coding - sending independent random codewords over each carrier - along with treating interference as noise is sum-capacity optimal. We now proceed to describe the model and our result formally.

The class of $F$-carrier parallel $K$-user discrete many-to-one interference channels considered can be represented by the set of $K$ (vector) inputs $\vec{X}_i = (X_i^{(1)},X_i^{(2)}, \ldots, X_i^{(F)}) \in \mathcal{X}^{(1)} \times \mathcal{X}^{(2)}\ldots \mathcal{X}^{(F)}$ and a set of $K$ (vector) outputs $\vec{Y}_i=(Y_i^{(1)},Y_i^{(2)}, \ldots, Y_i^{(F)})\in \mathcal{Y}^{(1)}\times \mathcal{Y}^{(2)}\ldots \mathcal{Y}^{(F)}$ for $i=1,2,\ldots,K.$ The $k$th component of output, $Y_i^{(k)}$ is determined by the $k$th component of the inputs $X_{j}^{(k)}, j=1,2,\ldots,K$ by any member of the class of discrete memoryless many-to-one interference channels defined earlier in Section \ref{sec:sysmod}. Note that the variables $\vec{V}=(V^{(1)}, V^{(2)}, \ldots,V^{(F)})$ and the functions $\vec{f}_1=(f_1^{(1)},f_1^{(2)},\ldots,f_1^{(F)})$ are used in defining the output are Receiver $1$, with $V^{(k)}, f_1^{(k)}$ used in defining the $k$th component of the output. There are $K$ messages in the system $W_1,W_2,\ldots,W_K$. The definition of a code of length $T$, the probability of error, the corresponding rate of the code, the capacity region and the sum-capacity are defined similar to the Section \ref{sec:sysmod}. The only difference is that, here, we restrict our study to constraints on the codewords, which ensure that either \begin{itemize}

\item $H\left(Y_i^{(k)T}\right)$ and $H\left(Y_i^{(k)T}|X_{i}^{(k)T}\right)$ exist for $k=1,2,\ldots,F$ and $i=1,2,\ldots,K$, or
\item $h\left(Y_i^{(k)T}\right)$ and $h\left(Y_i^{(k)T}|X_i^{(k)T}\right)$ exist for $k=1,2,\ldots,F$ and $i=1,2,\ldots,K$,
\end{itemize}
where, as before, $T$ denotes the length of the codeword. The class of parallel discrete memoryless many-to-one interference channels considered here is, in fact, a special case of the class of channels defined in Section \ref{sec:sysmod}.
Before we proceed to our main result, we provide a parallel extension of Lemma \ref{lemma:extremal}.
\begin{lemma}
\label{lemma:separable}
{\it
Consider random sequences $\vec{A}^{T},\vec{B}^{T},\vec{C}^{T}$ such that $\vec{A}=(A^{(1)},A^{(2)},\ldots,A^{(F)})$ and similarly, $\vec{B}$ and $\vec{C}$ represent $F$-dimensional vectors/tuples. $\vec{B}^T,\vec{C}^T$ are generated as $$p\left(\vec{B}^{T},\vec{C}^{T},\vec{A}^{T}\right)=p\left(\vec{A}^{T}\right)\displaystyle\prod_{k=1}^{F}\prod_{\tau=1}^{T}p\left(B^{(k)}(\tau)|A^{(k)}(\tau)\right)p\left(C^{(k)}(\tau)|B^{(k)}(\tau)\right),$$ where $p\left(B^{(k)}(\tau)|A^{(k)}(\tau)\right)=p_{B|A}^{(k)}, p\left(C^{(k)}(\tau)|B^{(k)}(\tau)\right)= p_{C|B}^{(k)}$, for all $\tau=1,2,\ldots,T, k=1,2,\ldots,F$. Note that the sequence $\vec{A}^{T}$ does \emph{not} have to be generated in an i.i.d fashion. The sequences $\vec{B}^{T},\vec{C}^{T}$ can be interpreted to be outputs of a physically degraded parallel ($F$-carrier) discrete memoryless broadcast channel whose (vector) input sequence is the $\vec{A}^T$. Also note that $\vec{A}^{T} \rightarrow \vec{B}^{T} \rightarrow \vec{C}^{T}$. Now, if $H(B^{(k)T}),H(C^{(k)T})$ exist for all $k=1,2,\ldots,F$, then 
$$H(\vec{B}^{T})-H(\vec{C}^{T}) \leq \sum_{k=1}^{F}\sum_{\tau=1}^{T} \left(H\left(B^{(k)}(\tau)\right)-H\left(C^{(k)}(\tau)\right)\right)$$
If $h(B^{(k)T}),h(C^{(k)T})$ exist for $k=1,2,\ldots,F$, then 
$$h(\vec{B}^{T})-h(\vec{C}^{T}) \leq \sum_{k=1}^{F} \sum_{\tau=1}^{T} \left(h\left(B^{(k)}(\tau)\right)-h\left(C^{(k)}(\tau)\right)\right)$$
Further, suppose that the channels describing $\vec{B}$ and $\vec{C}$ are {Gaussian} parallel broadcast channels, with an average covariance constraint on the input corresponding to each carrier, i.e., with 
$$ E\left[\frac{1}{T}\sum_{\tau=1}^{T} ||{A}^{(k)}(\tau){A}^{(k)}(\tau)^{\dagger}\right] = \Gamma^{(k)}$$ for some set of covariance matrices $\Gamma^{(k)},k=1,2,\ldots,F$, then i.i.d circularly symmetric Gaussian distribution on $\vec{A}^{T}$, with each $A^{(k)}$ independent of $A^{(\overline{k})}$ for $k \neq \overline{k}$ maximizes the quantity $h(\vec{B}^{T})-h(\vec{C}^{T})$.}
\end{lemma}
The proof of the lemma, which is similar to the proof of Lemma \ref{lemma:extremal}, is placed in Appendix \ref{app:separable}.
\begin{corollary}
{\it 
Consider a parallel discrete memoryless many-to-one interference channel, where each of these parallel channels satisfy the noisy interference conditions of Theorem \ref{thm:noisy}, i.e., where 
$V^{(k)}$ is a degraded version of $Y_{\Kappa}^{(k)}$ w.r.t $\mathcal{X}_{\Kappa}^{(k)}$. Then separate random coding over each of the parallel carriers and treating interference as noise achieves sum-capacity. The sum-capacity $C_{\Sigma}$ can therefore be written as 
$$ C_{\Sigma} = \sum_{k=1}^{F}\sum_{i=1}^{K} I(X_i^{(k)}; Y_i^{(k)})$$}
\end{corollary}
The proof of the above corollary omitted here, since it is almost identical to the proof of Theorem \ref{thm:noisy}, with Lemma \ref{lemma:separable} used in the proof, in place of Lemma \ref{lemma:extremal}. Since the class of parallel many-to-one interference channels described above is a special case of the class of many-to-one channels described in Theorem \ref{thm:noisy}, the optimality of random coding and treating interference as noise simply follow from the theorem - the additional insight of the above corollary is that the optimal input distribution involves the principle of separate coding, i.e., in the optimal input distribution, $X_{i}^{(k)}$ is independent of $X_i^{(\overline{k})}$ for $k \neq \overline{k}$ for all $i \in \Kappa$. The above corollary automatically implies that a set of parallel Gaussian many-to-one interference channels, each of which satisfies the noisy interference condition, is separable.
\section{On Interference Alignment In Noisy Interference Regimes and other Insights from Deterministic Many-to-One Interference Channels}
\label{sec:deterministic}

The deterministic many-to-one interference channel is the channel as described earlier, where $\mathcal{X}_i, \mathcal{Y}_i, \mathcal{V}$ are all finite and
\begin{equation} H(Y_i|X_i)=H(V|X_{\Kappa_1})=0, \forall i \in \Kappa_1, \label{eq:deterministic}\end{equation}
for all possible distributions on the input. Note that in this channel, the outputs can be uniquely determined from the set of inputs of the channel. Also note that this class of channels captures the deterministic framework proposed by \cite{Avestimehr_Diggavi_Tse}, and studied in the many-to-one interference channel setting in \cite{Bresler_Parekh_Tse}. We next proceed to understand the idea of interference alignment via random codes in this deterministic framework.

\subsection{Discussion : Why is explicit interference alignment not required in the noisy interference regime?}
\label{sec:discussion}
For the deterministic many-to-one channel defined in (\ref{eq:deterministic}), we present the noisy interference regime in a slightly different form which leads to interesting interpretation later in this section.
\begin{corollary}
\label{cor:noisy}
{\it
In the many-to-one interference channel, if there exists a function $q$ such that $V=q(Y_{\Kappa_1})$, then the sum-capacity is given by 
$$ C_{\Sigma} = \max_{\prod_{i \in \Kappa} p(X_i)}H(Y_1)+H(Y_{\Kappa_1}|V)$$}
\end{corollary}
\begin{proof}
The proof follows from setting $H(Y_i|X_{i})=H(V|X_{\Kappa_1})=0$ for all $i \in \Kappa_i$ in Theorem \ref{thm:noisy}, and noting that 
$$\sum_{i\in\Kappa}H(Y_i)-H(V)=H(Y_{\Kappa_1})-H(V)=H(Y_{\Kappa_1}|V)$$
since $Y_i$ is independent of $Y_{j}$ for $i \neq j$, and $V$ is a deterministic function of $Y_{\Kappa_1}$.
\end{proof}

\subsubsection*{Remark 1}
The above class of channels can be considered to be \emph{weak} interference channels, because the condition of the result implies that the effective interference at Receiver $1$ must be reconstructible from signals received at all receivers $i \neq 1$. The fact that random coding and treating interference as noise in the weak many-to-one interference channels is optimal can also be verified in the class of deterministic many-to-one interference channels studied by Bresler, Parekh and Tse \cite{Bresler_Parekh_Tse}.

For a better understanding of the noisy interference regime, let us take a closer look at the idea of interference alignment over the deterministic $3$-user many-to-one channel. Over this channel, consider a random coding scheme of length $T$, such that it generates $2^{TR_2}$ typical sequences of $X_{2}^T$ and $2^{TR_3}$ typical sequences of $X_{3}^T$. Since messages in the system are distributed uniformly, this means that $H(X_{i}^T)=TR_i,i=2,3$. Now, if $(R_2,R_3)$ lie in the achievable rate region (with these distributions) of the multiple access channel formed with inputs $X_2,X_3$ and output $V$, then the sequences $X_{2}^T$ and $X_{3}^T$ are invertible (i.e., decodable) from $V^T$. This is the case because random coding is optimal in the multiple access channel. In other words, the sequences $X_{2}^T,X_{3}^T$ are \emph{not} aligned; in fact, they are resolvable from $V^N$, and each $V^N$ sequence therefore corresponds to a unique $X_{2}^T,X_{3}^T$ sequence pair with high probability. Such codewords would satisfy
\begin{eqnarray*}H(V^T) &\approx& T(R_2+R_3)=H(X_{2}^T)+H(X_{3}^T)\\&=&H(X_{2}^T,X_{3}^T) \end{eqnarray*}
The approximation sign is used above rather than equality, since the comparison is in an asymptotic sense.
Contrary to the above scenario, if $(R_2,R_3)$ lie \emph{outside} the rate region achieved with these distributions in the multiple access channel, then even with random coding, the sequences $X_{2}^T,X_{3}^T$ align. In particular, in the noisy interference regime, $V=q(Y_2,Y_3)$ and in achievable scheme, $R_i=H(Y_i)$, which imply that 
\begin{eqnarray*}H(V^T) < T(R_2+R_3)= H(Y_{2}^T,Y_{3}^T) \end{eqnarray*}
as long as $q$ is a non-invertible function (The case of $q$ being invertible falls in the class of the resolvable interference regime discussed later in this section). Note that the above condition holds in the noisy interference regime, \emph{for every possible input distribution.} Thus, explicit alignment by an appropriate choice of input distribution or using multi-letter based structured coding is not required, and random codes automatically align interference in the above channel. The noisy nature of the interference can be explained by the insights of reference \cite{Wu_Xie_AEP}, which noted that on a single-user channel, if a random code of a rate higher than a user's capacity is used, then the signal loses structure in the sense that the output satisfies an equipartition property independent of the codebook used. This loss of structure can be used to explain the noisy nature of the interference, on noting that for any given input distribution, with random coding, $R_i = H(Y_i) \geq I(V;X_i|X_{\Kappa-\{i\}})=H(V|X_{\Kappa - \{i\}})$ in the noisy interference regime; in other words, from the perspective of a receiver with output $V$, the rate of transmission of the user is higher than the corresponding user's mutual information and thus the interfering signal loses any structure imposed by its codebook. The additional insight here is that, if the rate of each incoming signal at a receiver is higher than that user's mutual information, then, not only does the signal lose its structure, but the multiple signals also align. In fact, we will argue later in this section (Section \ref{sec:discussion}) that the extent of alignment is also the maximum possible in this case. These insights carry through to the Gaussian case as well, where, if each user $i \in \Kappa_1$ transmit at rates corresponding their single-user capacity, then the interference is noisy, and is aligned - even by circularly symmetric Gaussian distributions at all inputs.
\subsection{Resolvable Interference Regime for Deterministic Many-to-one Interference Channels}
\label{sec:resolvable}
In this class of many-to-one interference channels, interference cannot be aligned since the multiple interferers at the first receiver are resolvable. Since alignment is not possible, random coding achieves the capacity region. We first show inner and outer bounds for the deterministic many-to-one interference channel, respectively, in Theorems \ref{thm:outerbound} and \ref{thm:ach}. We then find conditions where these bounds are tight to define the resolvable interference regime in Corollary \ref{cor:resolvable}. The bounds and the regime are all defined in terms of auxiliary variables $U_{\Kappa_1}=(U_2,U_3,\ldots,U_{K})$ such that
\begin{itemize}
\item $U_i$ is a deterministic function of $X_i$ and 
\item $V$ is a deterministic function of $(U_2,U_3,\ldots,U_K)$. 
\end{itemize}
Note that $U_i=X_i$ provides a trivial assignment of auxiliary variables $U_i$. However, the bounds can be optimized over all possible choices of $U_i$ satisfying these properties. We now proceed to describe an outer bound on the capacity region of the many-to-one interference channel.

\begin{theorem}
\label{thm:outerbound}
{\it 
The capacity region of the deterministic many-to-one interference channel lies in the convex hull of the following region, over all possible product distributions $\prod_{i\in\mathcal{K}} p_{X_i}(x_i).$ 
\allowdisplaybreaks{
\begin{eqnarray}
R_1 &\leq& H(Y_1|V) \label{eqn:outerbound1}\\
R_i &\leq& H(Y_i), i\in \Kappa_1 \label{eqn:outerbound2}\\
R_1+\sum_{i \in \mathcal{S}}R_i  &\leq& H(Y_1|{U}_{\mathcal{S}^c})+H(Y_{\mathcal{S}}|V, {U}_{\mathcal{S}^{c}}), \forall \mathcal{S} \subseteq \Kappa_1, \label{eqn:outerbound3}
\end{eqnarray}
}
where $\mathcal{S}^{c}$ represents the complement of $\mathcal{S}$ w.r.t $\Kappa_1$. }
\end{theorem}
The above outer bound is proved in Appendix \ref{app:outerbound}.

We now describe below, a rate achievable in general, in the deterministic many-to-one interference channel using a random coding scheme which does not align interference. The achievable scheme is similar to the one presented in \cite{Cadambe_Jafar_Vishwanath} in the context of the deterministic $Z$ channel.

\begin{theorem}
\label{thm:ach}
{\it 
For the deterministic many-to-one interference channel, the convex hull over all product input distributions $\prod_{i \in \Kappa} p_{X_i}(x_i)$, of the following rate region is achievable.
\begin{eqnarray}
R_1 &\leq& H(Y_1|V) \label{eqn:ach1}\\
R_i &\leq& H(Y_i), i\in \Kappa_1 \label{eqn:ach2}\\
R_1+\sum_{i \in \mathcal{S}}R_i  &\leq& H(Y_1|{U}_{\mathcal{S}^c})+ \sum_{i \in \mathcal{S}} H(Y_i|{U}_{i}), \nonumber \\&&\forall \mathcal{S} \subseteq \Kappa_1, \mathcal{S}=\Kappa_1-\mathcal{S}. \label{eqn:ach3} 
\end{eqnarray}
}
\end{theorem}
The proof is placed in Appendix \ref{app:ach}.
It should be noted that the above achieved rate region is loose, in general, with respect to the bound of Theorem \ref{thm:outerbound}. However, if $H({U}_{\mathcal{S}}|V)=0$ for all $\mathcal{S} \subseteq \Kappa_1$, then the achieved rate region can be verified to be optimal by comparing (\ref{eqn:outerbound1})-(\ref{eqn:outerbound3}) with (\ref{eqn:ach1})-(\ref{eqn:ach3}). We state this formally below. 
\begin{corollary}
\label{cor:resolvable}
{\it 
Consider a many-to-one interference channel where $U_{\Kappa_1}$ is invertible from $V$, i.e., $H(U_{\Kappa_1}|V)=0$ for all possible input distributions.
Then, the capacity region of the many-to-one interference channel is given by (\ref{eqn:ach1})-(\ref{eqn:ach3}).}
\end{corollary}
\begin{proof}
Note that it is sufficient to show that the right hand sides of (\ref{eqn:ach3}) and (\ref{eqn:outerbound3}) are equal. We show this below.
\allowdisplaybreaks{
\begin{align}
&H(Y_1|{U}_{\mathcal{S}^c})+ \sum_{i \in \mathcal{S}} H(Y_i|U_{i}) \nonumber \\
&= H(Y_1|{U}_{\mathcal{S}^c})+ H(Y_{\mathcal{S}}|{U}_{\mathcal{S}}) \label{eq:indep3}\\
&= H(Y_1|{U}_{\mathcal{S}^c})+ H(Y_{\mathcal{S}}|{U}_{\mathcal{S},} {U}_{\mathcal{S}^{c}})\label{eq:indep4}\\
&= H(Y_1|{U}_{\mathcal{S}^c})+ H(Y_{\mathcal{S}}|V,{U}_{\mathcal{S}}, {U}_{\mathcal{S}^{c}})\\
&= H(Y_1|{U}_{\mathcal{S}^c})+ H(Y_{\mathcal{S}}|V,{U}_{\mathcal{S}^{c}}). \label{eq:invertibility}
\end{align}}
In (\ref{eq:indep3}),(\ref{eq:indep4}), we have used the fact that $(Y_{i},U_{i})$ is independent of $Y_{j},U_{j}$ for $i \neq j, i,j \in \mathcal{\Kappa}_1$. We have also used the fact that ${U}_{\mathcal{S}}$ is invertible from $V$ in the final equation above.
\end{proof}

Note that $U_i$ can be interpreted as the effective interference caused by Transmitter $i \neq 1$ at Receiver $1$. Also, note that with any achievable coding scheme in this channel, Receiver $1$ can decode $X_1^T$, and because of (\ref{eqn:invertible}), invert $V^T$ in this channel. The condition that ${U}_{\Kappa}$ is invertible from $V$ means that all the interfering signals $U_i^T$ are resolvable at the first receiver, and hence alignment is precluded irrespective of the coding scheme used. Therefore, not surprisingly, the random coding achievable scheme of Theorem \ref{thm:ach} is optimal for this class of channels.

It must be noted that, the achievable schemes of Theorem \ref{thm:ach} and Corollary \ref{cor:noisy} are both \emph{different} random coding schemes. The schemes differ, in particular, in the decoding procedure at Receiver $1$ and hence achieve different rates. In the achievable scheme of Theorem \ref{thm:ach}, the rate region achieved is with Receiver $1$ picking the sequence $X_1^T$ such that $Y_1^T,{U}_\Kappa^T,X_1^T$ are jointly typical (See Appendix \ref{app:ach}). In contrast, in the decoding scheme for Corollary \ref{cor:noisy},  the sequence $X_1^T$ is decoded as the one such that $(Y_1^T,X_1^T)$ are jointly typical, i.e., the interference is treated as noise.

\subsection{Discussion : How many bits of additional rate can interference alignment provide?}
\label{sec:discussion1}
The achievable scheme of Theorem \ref{thm:ach} does not involve interference alignment and is therefore optimal, only when alignment is precluded on the many-to-one interference channel. The resolvability condition of the channel described in Corollary \ref{cor:resolvable} above is precisely one where alignment is precluded. However, in general, if the resolvability condition is not satisfied, then on comparing (\ref{eqn:outerbound1})-(\ref{eqn:outerbound3}) with (\ref{eqn:ach1})-(\ref{eqn:ach3}), we can conclude that an additional rate of $\Delta_{\mathcal{S}}=H(Y_{\mathcal{S}}|V,{U}_{\mathcal{S}^c})-H(Y_{\mathcal{S}}|{U}_{\mathcal{S}})$ should be achieved by alignment for the users belonging in $\mathcal{S}$ for the outer bound to be tight. It is not clear whether this additional rate can be achieved at all, in general, or whether the outer bound is loose. However, the results on the noisy interference regime imply that if the many-to-one interference channel is weak, then this additional rate can be achieved using interference alignment via random coding, and the outer bound is tight in a sum-capacity sense (Compare expression of Corollary \ref{cor:noisy} with (\ref{eqn:outerbound3})). In other words, the extent of alignment is optimal in the sense that the additional rate benefit provided by alignment is the maximum possible. If the channel is not weak, then $\Delta_{\mathcal{S}}$ can be interpreted as a bound on the amount of additional rate that can be obtained via alignment for the users in $\mathcal{S}\subseteq \mathcal{K}_1$.

\subsection{An open question : When does a channel have a single-letter capacity characterization?}
A clear open problem motivated by this work is a capacity characterization of deterministic many-to-one interference channels. This question is particularly intriguing because it is not clear whether the channel allows a single-letter capacity characterization. Previous works on approximating the capacity of the channel motivate the need of structured (lattice) coding based achievable schemes \cite{Bresler_Parekh_Tse}. It has been discussed in \cite{Philosof_Zamir} that for channels where structured codes are necessary, single-letter characterizations may not exist. This is because coding schemes such as linear and lattice codes introduce structure as correlations in multiple uses of the channel. Interestingly, single-letter based lattice coding schemes (i.e., single-dimensional lattices) are shown to suffice for a degrees of freedom characterization of almost all interference channels in \cite{Motahari_DOFint}; however, it has been argued multi-dimensional lattices are useful for finer characterizations of capacity \cite{Jafarian_Jose_Vishwanath_algebraic}. Thus, the question of existence of single-letter characterizations of interference channels, and more general wireless networks remains wide open. The issue of existence of single-letter capacity characterizations also appears in several broadcast channel scenarios. The study of degrees of freedom of compound broadcast networks \cite{Weingarten_Shamai_Kramer, Gou_Jafar_Wang} suggests the possibility of alignment and structured coding in the channel, whereas, for certain degraded settings, the broadcast (multicast) channel allows single-letter capacity characterizations (See \cite{Nair_El_Gamal} and references therein).
Thus, an important open question in network information theory is a better understanding of structured codes, and its impact on capacity characterizations of discrete memoryless channels. 

\section{Conclusion}
We generalize the noisy interference regimes, previously shown in average-power constrained SISO Gaussian interference channels, to the discrete memoryless many-to-one interference channel. In this noisy interference regime, random coding at all transmitters and treating interference as noise at the receiver which faces interference achieves sum-capacity on the many-to-one interference channel. Our generalization enables extension of the noisy interference regimes to the Gaussian MIMO and parallel many-to-one interference channels and the fading Gaussian many-to-one interference channels without CSIT. Unlike previous results which consider an average power constraint on the inputs, we also show that treating interference as noise is optimal in the Gaussian many-to-one interference channel, even if the inputs are constrained to come from fixed finite constellations such as QAM or PSK.
Through the lens of interference alignment, we obtain a better understanding of why random (Gaussian) codewords are sufficient to achieve capacity in the noisy interference regime in the Gaussian interference channels. In particular, we argue that if users transmit, using random coding, at rates higher than the interfering link's mutual information, then the interference is noisy and the extent of alignment is maximum. Such alignment hence obviates the need for techniques such as structured (lattice) codes which have been shown to be immensely useful in other regimes. We also show that for deterministic many-to-one interference channels, if the interferers are resolvable at the receiver facing interference, random coding achieves capacity since interference alignment is precluded. While we are able to provide single-letter characterizations for certain classes of channels in this paper, the question of the existence of single letter characterizations for wireless networks, in general, remains open.

\appendices
\section{Equivalence of Our Noisy Interference Regime and that of \cite{Sreekanth_Veeravalli,Motahari_Khandani, Shang_Kramer_Chen} for Single-Antenna Gaussian Many-to-One Interference Channels}
\label{app:equivalence}
We have already shown in the introduction that the noisy interference regime of \cite{Sreekanth_Veeravalli,Motahari_Khandani, Shang_Kramer_Chen} is included in the regime described by our result. We here show that our noisy interference regime is no larger than the regime described in \cite{Sreekanth_Veeravalli,Motahari_Khandani,Shang_Kramer_Chen}. In particular, we will show here that, if the conditions of (\ref{eq:classicalnoisy}) are \emph{not} satisfied, then, the effective interference at Receiver $1$ is {not} a degraded version of the set of all signals at the other receivers. In particular, we will show that, if 
\begin{equation}
\sum_{j=2}^{K} \frac{|H_{1j}|^2}{|H_{jj}|^2} > 1, \label{eq:notnoisy}
\end{equation}
then, for any possible set of values $\rho_i=E[Z_1 Z_i], i=2,3,\ldots,K$ with, there exists an input distributions on $X_i,i=1,2,\ldots,K$ such that 
\begin{equation} I(X_2,X_3,\ldots,X_K; {V}|Y_2,Y_3,\ldots,Y_K) > 0, \label{eq:notnoisy2}\end{equation}
where $V=\sum_{i=2}^{K}H_{1i}X_i+Z_1$.
To see this, note that, since $(Z_2,Z_3,\ldots,Z_K)$ are a set of mutually independent Gaussian random variables, the correlation matrix of $(Z_1,Z_2,\ldots,Z_K)$ can be written as 
$$\left[\begin{array}{cccccc} 1 &\rho_1&\rho_2& \rho_3&\ldots&\rho_K\\
		       \rho_1 & 1 &   0  &0   &\ldots& 0 \\
			\rho_2 & 0 & 1& 0&\ldots&0 \\
			\rho_3 & 0 & 0&1&\ddots&\vdots \\
			\vdots&\vdots&\vdots&\ddots&\ddots&0\\
			\rho_K&0&0&\ldots& 0&1 \end{array} \right].$$
Note that the above matrix has to be a positive semidefinite matrix, whose determinant is non-negative. This implies that 
$$\sum_{i=2}^{K}\rho_{i}^2 \leq 1$$
(\ref{eq:notnoisy}) and the above equation together imply that there exists $i_0 \in \{2,3,\ldots,K\}$ such that
$\frac{|H_{1i_0}|}{|H_{i_0i_0}|} > |\rho_{i_0}|$.  We are now ready to provide the input distribution for which (\ref{eq:notnoisy2}) is satisfied. Consider the input distribution where $X_i=0$ deterministically for all $i \in \{2,3,\ldots,K\}-\{i_0\}$, and $X_{i_0}$ is a circularly symmetric zero-mean Gaussian random variable having some positive (non-zero) variance. Then, we can write 
\begin{align}&I(X_2,X_3,\ldots,X_K; V|Y_2,Y_3,\ldots,Y_K) \nonumber \\&= I(X_{i_0}; H_{1i_0}X_{i_0}+Z_1|Z_1,Z_2,\ldots,Z_{i_0-1},H_{i_0i_0}X_{i_0}+Z_{i_0},Z_{i_0+1},\ldots,Z_{K}) \nonumber\\ 
 &\geq I(X_{i_0}; H_{1i_0}X_{i_0}+Z_1|H_{i_0i_0}X_{i_0}+Z_{i_0}) \label{eq:degraded1}\\&\geq 0 \label{eq:degraded2} \end{align}
where (\ref{eq:degraded1}) follows from the independence of $X_{i_0}$ and $Z_i, i=\{2,3\ldots,K\}$. In (\ref{eq:degraded2}), the equality is satisfied (for Gaussian $X_{i0}$) only if $|\rho_{i_0}|=\frac{|H_{1i_0}|}{|H_{i_0i_0}|}$. However this condition is not satisfied for our choice of $i_0$, and we have, for this input distribution, $I(X_2,X_3,\ldots,X_K; V|Y_2,Y_3,\ldots,Y_K) > 0$, i.e., the interference at Receiver $1$ is \emph{not} a degraded version of the outputs at the other receivers. This shows the equivalence between our characterization of the noisy interference regime and the characterization in \cite{Sreekanth_Veeravalli, Motahari_Khandani,Shang_Kramer_Chen}. It must be noted that we do not claim if (\ref{eq:notnoisy}) is satisfied, that the interference is not noisy - we only claim that it lies outside our characterization of the noisy interference regime. Whether the noisy interference regime characterized by us can be expanded is an interesting open question.
\section{Proof of Lemma \ref{lemma:extremal}}
\label{app:extremal}
From the Markov chain property on $A,B,C$, we can write 
\begin{eqnarray*}
I(A^{T};B^{T}) &=&I(A^{T};B^{T},C^{T})\\
I(A^{T};B^{T}) &=&I(A^{T};C^{T})+I(A^{T};B^{T}|C^{T})\\
\Rightarrow  h(B^{T})-h(C^{T}) &=&  h(B^{T}|C^{T})+ h(B^{T}|A^{T})-h(C^{T}|A^{T})-h(B^{T}|C^{T},A^{T})\\
&\stackrel{(a)}{\leq}& \sum_{\tau=1}^{T}  h(B(\tau)|C(\tau))+ h(B(\tau)|A(\tau))-h(C(\tau)|A(\tau))-h(B(\tau)|C(\tau),A(\tau))\\
&=& \sum_{\tau=1}^{T}  h(B(\tau)|C(\tau))+ h(B(\tau)|A(\tau))-h(C(\tau),B(\tau)|A(\tau))\\
&=& \sum_{\tau=1}^{T}  h(B(\tau)|A(\tau))+ h(B(\tau), C(\tau))-h(C(\tau))-h(C(\tau),B(\tau)|A(\tau))\\
&=& \sum_{\tau=1}^{T}  h(B(\tau)|A(\tau))+ I(B(\tau), C(\tau);A(\tau))-h(C(\tau))\\
&=& \sum_{\tau=1}^{T}  h(B(\tau))-I(B(\tau);A(\tau))+ I(B(\tau), C(\tau);A(\tau))-h(C(\tau))\\
&\stackrel{(b)}{=}& \sum_{\tau=1}^{T}  h(B(\tau))-h(C(\tau)),
\end{eqnarray*}
where in $(a)$, we have used the chain rule and the fact that conditioning cannot increase differential entropy in the first two terms, and the following facts in the final two terms above  $$p(C^{T}|A^{T})=\prod_{\tau=1}^{T} p(C(\tau)|A(\tau)), $$  $$p(B^{T}|C^{T},A^{T})=\frac{p(B^{T},C^{T}|A^{T})}{p(C^{T}|A^{T})}=\prod_{\tau=1}^{T} \frac{p(B(\tau),C(\tau)|A(\tau))}{p(C(\tau)|A(\tau))} = \prod_{\tau=1}^{T} p(B(\tau)|C(\tau),A(\tau)).$$
In $(b)$, we have used the Markov chain property on $A,B,C$ which implies that $I(A(\tau);B(\tau),C(\tau))=I(A(\tau);B(\tau)).$ Now, if $p_{B|A},p_{C|B}$ are both Gaussian, then the fact that the Gaussian input distribution on the $A(\tau)$ with the appropriate covariance matrix maximizes $h(B(\tau)|C(\tau))$ in step $(a)$, combined with the convexity of entropy implies that i.i.d Gaussian distribution on the input maximizes $h(B^{T})-h(C^{T})$.

\section{Proof of Theorem \ref{thm:noisy}}
\label{app:noisy}
Achievability of the required rate follows trivially from typical set decoding arguments. We prove the converse here.
We prove the converse for the case the differential entropies of ${Y}_k$ and $Y_k|X_k$ exist. The proof for the case the corresponding entropy terms exist, rather their differential entropies, is essentially identical, with the mutual information expressed in terms of the entropy of the variables.
Consider any coding scheme of length $T$ achieving rates $R_i$ for user $i$. Then, from Fano's inequality, for any $\epsilon > 0$, we can write 
\begin{eqnarray}
T(R_1-\epsilon/K) &\leq& I(Y_1^T;X_1^T) \nonumber\\
&\leq& h(Y_1^T)-h(Y_1^T|X_1^T)\nonumber\\
&=& h(Y_1^T)-h(V^T)\nonumber\\
&\leq& \sum_{\tau=1}^{T}h(Y_1(\tau))-h(V^T),\label{eq:R1}
\end{eqnarray}
where we have used (\ref{eqn:invertible}) above. Now, note that since the capacity of the channel only depends the marginal distributions $p_{V|X_{\Kappa_1}}, p_{Y_i|X_i}, i \in \Kappa_1$, we can make $V$ a physically degraded version of $Y_{\Kappa_1}$ so that $X_i^T \rightarrow Y_i^T \rightarrow V^T$ without changing the capacity. We can now write, for $i \in \mathcal{K}_1$. 
\begin{eqnarray}
T(R_i-\epsilon/K) &\leq& I(Y_i^T;X_i^T)\nonumber\\
&=& h(Y_i^T) - h(Y_i^T| X_i^T)\nonumber\\
&=& h(Y_i^T) - \sum_{\tau=1}^{T}h(Y_i(\tau)| X_i(\tau))\nonumber\\
\Rightarrow T\left(\sum_{i \in \mathcal{K}_1} R_i - (K-1)\epsilon/K\right) &\leq& \sum_{i \in \mathcal{K}_1} \bigg(h(Y_i^T) - \sum_{\tau=1}^{T}h(Y_i(\tau))|X_i(\tau))\bigg)\nonumber\\
 &\leq& h\left(Y_{\Kappa_1}^T\right) - \sum_{i \in \mathcal{K}_1} \sum_{\tau=1}^{T} h\left(Y_i(\tau)|X_i(\tau)\right) \label{eq:independent1}.
\end{eqnarray}
where the final equation follows from the fact that $Y_{i}^T$ is independent of $Y_j^T$ for all $i \neq j$.
Adding (\ref{eq:R1}) and (\ref{eq:independent1}), we get
\allowdisplaybreaks{
\begin{align}
  & T\left(\sum_{i \in \mathcal{K}} R_i - \epsilon\right)\\ &\leq \sum_{\tau=1}^{T}h(Y_1(\tau)) - h(V^T)+h(Y_{\Kappa_1}^T)  - \sum_{i \in \mathcal{K}_1} \sum_{\tau=1}^{T} h(Y_i(\tau)|X_i(\tau)) \\
   &\leq \sum_{\tau=1}^{T}h(Y_1(\tau))+ \sum_{\tau=1}^{T} \big(h(Y_{\Kappa_1}(\tau)) - h(V(\tau))\big) - \sum_{i \in \mathcal{K}_1} \sum_{\tau=1}^{T} h(Y_i(\tau)|X_i(\tau)) \label{eq:fromextremal}\\
 &= \sum_{\tau=1}^{T}\bigg(h(Y_1(\tau))-h(V(\tau)) +\sum_{i \in \mathcal{K}_1} \big(h(Y_i(\tau)) - h(Y_i(\tau)|X_i(\tau))\big)\bigg) \\
 &= \sum_{\tau=1}^{T}\bigg(I(Y_1(\tau);X_1(\tau))+ \sum_{i \in {\Kappa_1}} I(X_i(\tau);Y_i(\tau))\bigg)\\
 &= \sum_{\tau=1}^{T}\sum_{i\in \Kappa} I(Y_i(\tau);X_i(\tau))\\
&\leq T \max_{\tau} \sum_{i\in \Kappa} I(X_i(\tau);Y_i(\tau))\\
&\leq T \max \sum_{i \in \Kappa} I(X_i;Y_i)
\end{align}
}
\noindent where (\ref{eq:fromextremal}) follows from the fact that $X_{\Kappa_1}^T \rightarrow Y_{\Kappa_1}^T \rightarrow V^T$ combined with the result of Lemma \ref{lemma:extremal}. The outer bound hence follows.

\section{Proof of Corollary \ref{cor:Gaussian}}
\label{app:Gaussian}
The condition that $\sum_{i \in \Kappa_1} \Lambda_{i} \preceq I_{N_1}$ implies that a Gaussian random $N_1$ dimensional vector $Z$ can be found so that $V=\sum_{i \in \Kappa_1} U_{i}+Z$ and $Y_1=X_1+V$, which means that $V$ is a degraded version of $U_{\Kappa_1}$. This fact, combined with the condition that $U_i$ is degraded version of $Y_i$ implies that $V$ is a degraded version of $Y_{\Kappa_1}$ as required by Theorem \ref{thm:noisy}. The optimality of random coding and treating interference as noise hence follows from the theorem.  Here, we only need to show that Gaussian inputs are optimal, when $\mathcal{X}_i = \mathbb{C}^{M_i}$, and there is a power constraint on the inputs. Consider any achievable scheme where $E\left(\frac{1}{T} \sum_{\tau=}^{T} X_i(\tau) X_i(\tau)^{\dagger} \right) = \Gamma_i$. Then, following the steps of the proof of Theorem \ref{thm:noisy}, we can derive equation (\ref{eq:fromextremal}), which is reproduced below.
\begin{eqnarray*}T\left(\sum_{i \in \Kappa} R_i - \epsilon \right) &\leq& \sum_{\tau=1}^{T}h(Y_1(\tau))+ \sum_{\tau=1}^{T} \big(h(Y_{\Kappa_1}(\tau)) - h(V(\tau))\big) - \sum_{i \in \mathcal{K}_1} \sum_{\tau=1}^{T} h(Y_i(\tau)|X_i(\tau)) 
\end{eqnarray*}
Here, we can use the Gaussian distribution to evaluate each of the entropy terms above. To see this, we invoke the convexity of entropy, and the fact that under a covariance matrix constraint, the Gaussian distribution maximizes entropy in the first term above. We can use Lemma \ref{lemma:extremal} which shows that the use of the Gaussian distribution to evaluate the second and third terms above, outer bounds the terms. The final entropy term is evaluated using Gaussian because of the definition of the channel. Thus, if we have a power constraint, we are restricted to the set of all Gaussian distributions on the input at Transmitter $i$ with covariance $\Gamma_i$, where $\textrm{tr}(\Gamma_i) \leq P_i$.

\section{Proof of Lemma \ref{lemma:separable}}
\label{app:separable}
\begin{proof}
The proof is similar to the proof of Lemma \ref{lemma:separable}. We only highlight the differences here. We have $\vec{A}^{T} \rightarrow \vec{B}^{T} \rightarrow \vec{C}^{T}$
\begin{eqnarray*} I(\vec{A}^{T};\vec{B}^{T})&=& I(\vec{A}^T; \vec{B}^{T} | \vec{C}^T)\\
\Rightarrow  h(\vec{B}^{T})-h(\vec{C}^{T}) &=&  h(\vec{B}^{T}|\vec{C}^{T})+ h(\vec{B}^{T}|\vec{A}^{T})-h(\vec{C}^{T}|\vec{A}^{T})-h(\vec{B}^{T}|\vec{C}^{T},\vec{A}^{T})\\
&\stackrel{(c)}{\leq}& \sum_{k=1}^{F}\sum_{\tau=1}^{T}  h\left(B^{(k)}(\tau)|C^{(k)}(\tau)\right)+ h\left(B^{(k)}(\tau)|A^{(k)}(\tau)\right) \nonumber\\&&-h\left(C^{(k)}(\tau)|A^{(k)}(\tau)\right)-h\left(B^{(k)}(\tau)|C^{(k)}(\tau),A^{(k)}(\tau)\right)\\
&=& \sum_{k=1}^{F}\sum_{\tau=1}^{T}  h\left(B^{(k)}(\tau)\right)-h\left(C^{(k)}(\tau)\right),
\end{eqnarray*}
where in $(c)$, we have used the chain rule and the fact that conditioning cannot increase differential entropy in the first two terms, and the definition of the channels $p_{B|A}^{(k)}, p_{C|B}^{(k)}$ in the final two terms above, similar to Lemma \ref{lemma:extremal}. The arguments for the derivation of the final step from $(c)$, and the optimality of Gaussian inputs if the channels are Gaussian, are identical to the proof of Lemma \ref{lemma:extremal} and hence, omitted here. 
\end{proof}

\section{Proof of Theorem \ref{thm:outerbound}} 
\label{app:outerbound}
The bounds (\ref{eqn:outerbound1}) and (\ref{eqn:outerbound2}) are trivial. We only need to show (\ref{eqn:outerbound3}). Consider any achievable scheme. Let $T$ be the length of the code. Consider any $\mathcal{S} \subseteq \Kappa_1$. Then, for $i \in \mathcal{S}$, from Fano's inequality we can write for any $\epsilon > 0$,
\allowdisplaybreaks{
\begin{eqnarray}
T(R_i-\epsilon) &\leq& H(Y_i^T),\\
\Rightarrow T\left(\sum_{i\in \mathcal{S}} R_i-|\mathcal{S}|\epsilon\right) 
&\leq& \sum_{i\in\mathcal{S}} H(Y_i^T)= H(Y_{\mathcal{S}}^T) = H(Y_{\mathcal{S}}^T|{U}_{\mathcal{S}^{c}}), \label{eqn:transient1}
\end{eqnarray}
where, in the final two equations, we have used the fact that $(Y_i,{U}_i)$ is independent of $(Y_j,{U}_j)$ for $j \neq i, (i,j) \in \Kappa_1$. Now, using Fano's inequality for $W_1$, we get
\begin{eqnarray}
T(R_1-\epsilon) &\leq& I(Y_1^T;X_1^T)\\
	&\leq& I(Y_1^T,{U}_{\mathcal{S}^c}^T;X_1^T)\\
	&=& I(Y_1^T;X_1^T|{U}_{\mathcal{S}^c}^T)\\
 &=& H(Y_1^T|{U}_{\mathcal{S}^c}^T)-H(Y_1^T|X_1^T,{U}_{\mathcal{S}^c}^T)\\
 &\leq& TH(Y_1|{U}_{\mathcal{S}^c})-H(V^T|{U}_{\mathcal{S}^c}^T) \label{eqn:transient2},
\end{eqnarray}
where, we have used (\ref{eqn:invertible}), and the convexity of the conditional entropy function above. Summing (\ref{eqn:transient1}) and (\ref{eqn:transient2}), we get 
\begin{eqnarray}
T\left(R_1+\sum_{i\in \mathcal{S}} R_i-K \epsilon\right) &\leq& TH(Y_1|{U}_{\mathcal{S}^c})+ H(Y_{\mathcal{S}}^T|{U}_{\mathcal{S}^c}^T)-H(V^T|{U}_{\mathcal{S}^c}^T)\\
&{\leq}& TH(Y_1|{U}_{S^c})+ H(Y_\mathcal{S}^{T}|V^T,{U}_{\mathcal{S}^c}^T) \label{eq:degraded}
\end{eqnarray}
In (\ref{eq:degraded}), the fact that for any arbitrary variables $A,B,C$, $H(A|C)-H(B|C) \leq H(A|B,C)$. Dividing the final equation by $T$ and taking $T \to \infty$, we get the desired bound.

\section{Proof of Theorem \ref{thm:ach}} 
\label{app:ach}

We provide a random coding achievable scheme along the lines of reference \cite{Cadambe_Jafar_Vishwanath}, which studied the deterministic $Z$ channel. Without loss of generality, let us assume that $\mathcal{W}_i = \{1,2,\ldots,2^{TR_{i}}\}$. Consider any product distribution $\prod_{i \in \Kappa} p_{X_i}(x_i)$. 
\subsubsection*{Encoding}
The first transmitter generates $2^{TR_{1}}$ independent codewords ${X}_1^{T}$ generating each element i.i.d according to $p_{X_1}(x_1)$. Let the generated sequences be denoted by $X_1^{T}(m), m \in \{1,2,\ldots,2^{TR_1}\}$. Then the message $W_{1}=m$ is encoded using $X_1^T(m)$.
Now consider Transmitter $i \in \mathcal{\Kappa}_1$. Note that $p_{X_i}(x_i)$ along with the channel induces $p_{X_i,U_{i},Y_i}(x_i,u_{i},y_i)$ from which marginal distributions $p_{U_{i}}(u_{i}),p_{Y_i}(y_i)$ and $p_{U_{i},Y_i}(u_{i},y_i)$ can be calculated. The transmitter generates $2^{T\Omega_i}$ sequences of $U_{i}^T$, each sequence generated independently and i.i.d according to $p_{U_{i}}(u_{i})$, where $\Omega_i > 0$. We denote the $m$th  sequence generated as $U_{i}^{T}(m)$, where $m \in \{1,2,\ldots,2^{T \Omega_i}\}$. The transmitter also generates $2^{TH(Y_i)}$ sequences of $Y_i^T$, each sequence generated independently and i.i.d according to $p_{Y_i}(y_i)$. These sequences of $Y_i^T$ are distributed uniformly into $2^{TR_i}$ bins. To encode the $m$th message, the transmitter picks a $Y_i^T$ sequence in the $m$th bin, such that it is jointly typical with $U_{i}^T(m_i)$ for some $m_i \in \{1,2,\ldots, 2^{T\Omega_i}\}$. If no such sequence is found, then an error is declared. Otherwise, the message is encoded using the $X_i^T$ which generates the $(U_{i}^{T},Y_i^{T})$ sequence picked. The existence of such a $X_i^{T}$ sequence is guaranteed, because the channel is deterministic and the pair $(U_{i}^{T}, Y_i^{T})$, by virtue of being jointly typical, has a non-zero probability of occurrence.

\emph{Remark:}  The encoding strategy at Transmitter $2$ is similar to the optimal coding strategy over the deterministic broadcast channel \cite{Cover_comments}.
\subsubsection*{Decoding Strategy }
Receiver $1$, on receiving $Y_1^T$, chooses the unique index $W_{1}=m$, such that $$(Y_1^T,{U}_2^T(m_2), U_{3}^T(m_3),\ldots,U_{K}^T(m_K), X_1^T(m))$$ is jointly typical, for some $m_i \in \{1,2,\ldots,2^{\Omega_i}\}$ for $i=2,3,\ldots,K.$  An error is declared if no such unique index $m$ is found. Receiver $i \in \mathcal{K}_1$ can decode $W_i$ using the bin-index of the received $Y_i^T$ sequence. Note that since the channel is deterministic, there are no errors at Receiver $i \neq 1$, if the encoding at Transmitter $i$ is successful.
\subsubsection*{Error Analysis and Achieved rate}
Since the coding scheme is symmetric over all messages, we will analyze the probability of error assuming $W_{i}=1$ is encoded at transmitter $i \in \Kappa$; because of symmetry in the coding scheme, the probability of error of encoding this set of messages gives the probability of error averaged over all messages.  Now, we divide the possible set of errors into two types : errors at Transmitter $i \in \mathcal{\Kappa}_1$ and errors at Receiver $1$.

For $i \in \Kappa_1$, let 
$$ E_i(m_j)=\left\{\begin{array}{c}\textrm{At Transmitter $i$, $Y_1^T$ belongs to the $m_j$th bin, $U_i^T=U_i^T(m)$} \\ \textrm{ for some } m \in \{1,2,\ldots,2^{T\Omega_i}\} \Rightarrow (U_i^T,Y_i^T) \notin \mathcal{A}_{\epsilon}(U_i,Y_i)\end{array}\right\}$$
where $\mathcal{A}_{\epsilon}(U_i,Y_i)$ represents the $\epsilon$-jointly typical set of $(U_i^T,Y_i^T)$ pairs. Note that $E_i(m_j)$ corresponds to the event that no jointly typical pair $Y_i^T,U_i^T$ was found in the $m_j$th bin at Transmitter $i$ when encoding $W_{i}=m_j$. The overall probability of error can now be expressed as 
$$ P_e^{T} = \sum_{i\in \Kappa_1} \Pr(E_i(1)) + \Pr(\textrm{Decoding error at receiver }1| E_2^{c}(1),E_3^{c}(1),\ldots, E_{K}^{c}(1)).$$
Now, consider Receiver $1$. Note that, at this receiver, the decoding procedure and hence the error events are very similar, in nature, to the errors that can occur over a multiple access channel (MAC), when the asymptotically optimal typical set coding procedure is used. The only difference is that, in this particular case, the receiver is only interested in one message, i.e, $W_1$, which reduces the number of possible error events as compared to the classical MAC. Now, given that the message $W_{i}=1$ is encoded at Transmitter $i$ for $i \in \Kappa$, and no errors occurred at the transmitter, a sequence $U_{i}^{T}$ is found at the transmitter such that it is jointly typical with a $Y_i^T$ sequence in the first bin. Let us assume, without loss of generality, that this sequence found is $U_{i}^T(1)$, i.e., $U_i^T(1)$ is used in the encoding procedure at Transmitter $i$, for $i \neq 1$. If $E_i^{c}(1)$ occurs for $i \neq 1$, then, because of the deterministic nature of the channel, the received sequence $Y_1^T$ is jointly typical with $U_2^T(1),U_3^T(1),\ldots, U_K^T(1),X_1^T(1)$. Errors can occur if $Y_1^T(1)$ is jointly typical with $U_2^T(m_2),U_3^T(m_3),\ldots,U_K^T(m_K),X_1^T(m_1)$ for some $m_1 \neq 1$, $m_i \in \{1,2,\ldots, 2^{T \Omega_i}\}$ for $i \in \Kappa_1$. We wish to evaluate the probability of occurrence of this event. Let us define 
$$ E_1(m_1,m_2,m_3,\ldots,m_K)=\{\textrm{ $Y_1^T, U_2^T(m_2),U_3^T(m_3), \ldots, U_K^T(m_K), X_1^T(m_1)$ is jointly typical}\}$$
By the union bound, we can bound the error at Receiver $1$ as
\begin{align*}
&\Pr(\textrm{Error at Receiver 1}|E_2^{c}(1),E_3^{c}(1),\ldots,E_K^{c}(1))\\ 
&\leq \sum_{m_1=2}^{2^{TR_1}} \sum_{m_i \in \{1,2,\ldots,2^{\Omega_i}\}, i \in \Kappa_1} \Pr(E_1(m_1,m_2,\ldots,m_K))\\
&= \sum_{m_1=2}^{2^{TR_1}}\sum_{\mathcal{S} \subseteq \Kappa_1} \sum_{ \substack{m_i\neq 1, i \in \mathcal{S},\\ m_i=1, i \in \mathcal{S}^{c}}} \Pr(E_1(m_1,m_2,\ldots,m_K))\\
&= \sum_{\mathcal{S} \subseteq \Kappa_1} \sum_{m_1=2}^{2^{TR_1}} \sum_{ \substack{m_i\neq 1, i \in \mathcal{S},\\ m_i=1, i \in \mathcal{S}^{c}}} \Pr(E_1(m_1,m_2,\ldots,m_K)
\end{align*}
Now, the overall error probability can be bounded as 
\begin{equation} P_e^{T} \leq  \sum_{\mathcal{S} \subseteq \Kappa_1} \sum_{m_1=2}^{2^{TR_1}}\sum_{ \substack{m_i\neq 1, i \in \mathcal{S},\\ m_i=1, i \in \mathcal{S}^{c}}} \Pr\left(E_1(m_1,m_2,\ldots,m_K)\right)+\sum_{i \in \Kappa_1} \Pr(E_i(1)).\label{eq:errors}\end{equation}
It has been shown in \cite{Gamal_Van} that if, for $i\neq 1$,
\begin{eqnarray}R_i &\leq& H(Y_i) \label{eq:parametric1}\\
R_i &\leq& \Omega_i + H(Y_i)-I(Y_i;U_i)\nonumber\\
&=& \Omega_i+H(Y_i|U_i) \label{eq:parametric2},
\end{eqnarray}
then, asymptotically as $T \to \infty$, the probability of $E_i(1)$ occurring vanishes.
Now, we estimate the remaining term in (\ref{eq:errors}) below. Let $\mathcal{A}_\epsilon$ denote the set of all $\epsilon$-jointly typical sequences of $(Y_1^T,U_2^T,\ldots,U_K^T,X_1^T).$ Then, for any $\mathcal{S} \subseteq \Kappa_1$, we can write
\begin{align*}
&\sum_{m_1=2}^{2^{TR_1}}  \sum_{ \substack{m_i\neq 1, i \in \mathcal{S},\\ m_i=1, i \in \mathcal{S}^{c}}}\Pr(E_1(m_1,m_2,\ldots,m_K)) \\
&\leq \sum_{m_1=2}^{2^{TR_1}}  \sum_{ \substack{m_i\neq 1, i \in \mathcal{S},\\ m_i=1, i \in \mathcal{S}^{c}}}\Pr\left((Y_1^T,U_2^T(m_2),U_3^T(m_3),\ldots,U_K^T(m_K), X_1^T(m_1)), \in \mathcal{A}_\epsilon\right)\\
&\leq \sum_{m_1=2}^{2^{TR_1}}  \sum_{ \substack{m_i\neq 1, i \in \mathcal{S},\\ m_i=1, i \in \mathcal{S}^{c}}} \sum_{(Y_1^T, U_{2}^T,U_{3}^T,\ldots,U_{K}^T,X_1^T)\in \mathcal{A}_{\epsilon}} \Pr\left(Y_1^T|U_{\mathcal{S}^c}^T\right) \Pr\left(X_1^T(m_1)=X_1^T\right) \prod_{j \in \Kappa_1} \Pr\left(U_i^T(m_j)=U_i^T\right) \\
&\stackrel{(a)}{\leq} \sum_{m_1=2}^{2^{TR_1}}  \sum_{ \substack{m_i\neq 1, i \in \mathcal{S},\\ m_i=1, i \in \mathcal{S}^{c}}} |\mathcal{A}_{\epsilon}| 2^{-TH(Y_1|U_{\mathcal{S}^c})+\epsilon} 2^{-TH(U_{\mathcal{K}_1})-TH(X_1))+K\epsilon} \\
&\leq \sum_{m_1=2}^{2^{TR_1}}  \sum_{ \substack{m_i\neq 1, i \in \mathcal{S},\\ m_i=1, i \in \mathcal{S}^{c}}}2^{-TI(Y_1; U_{\mathcal{S}},X_1|U_{\mathcal{S}^{c}})+(K+2)\epsilon}\\
&= 2^{TR_1} 2^{T \sum_{i \in \mathcal{S}} T \Omega_i} 2^{-TH(Y_1|U_{\mathcal{S}^{c}})+(K+2)\epsilon},
\end{align*}
where, in $(a)$, we have used the fact that $U_j$ is independent $U_{j^{'}}$ for $j \neq j^{'}$.
From the above equation, it can be concluded that, if, 
\begin{equation}
R_1 + \sum_{i \in \mathcal{S}} \Omega_i \leq H(Y_1| U_{\mathcal{S}^c}), \forall \mathcal{S} \subseteq \Kappa_1\label{eq:parametric3}
\end{equation}
then, asymptotically as $T \to \infty$, the probability of error at Receiver $1$ vanishes. Note that if $\mathcal{S}$ is the null-set, the above equation can be equivalently expressed as $R_1 \leq H(Y_1|V)$. Thus, if the rates $R_i, i \in \Kappa$ and the parameters $\Omega_i, i \in \Kappa_1$ satisfy (\ref{eq:parametric1}),(\ref{eq:parametric2}) and (\ref{eq:parametric3}), the rate-tuple $(R_1,R_2,\ldots,R_K)$ is achievable. Eliminating $\Omega_i, i \in \Kappa_1$ from these inequalities using Fourier-Motzkin elimination, we get the achieved rate region in the desired form.
}

\bibliographystyle{ieeetr}
\bibliography{Thesis}

\end{document}